\documentclass[runningheads]{llncs}
\usepackage[T1]{fontenc}
\usepackage{graphicx}
\usepackage{balance} % for balancing columns on the final page
\usepackage{multirow}
\usepackage{amsmath}
\usepackage{cleveref}
\usepackage[inline]{enumitem}
\usepackage{booktabs}
\usepackage{amssymb}

\begin{document}
\title{Optimal Allocations under Strongly Pigou-Dalton Criteria:
 Hidden Layer Structure \& Efficient Combinatorial Approach \thanks{This research is supported by Department of Science and Technology of Guangdong Province (Project No. 2021QN02X239) and Shenzhen Science and Technology Program (Grant No. 202206193000001, 20220817175048002).} } 
\titlerunning{Optimal Allocations under Strongly Pigou-Dalton Criteria}

\author{Taikun Zhu\inst{1}\orcidID{0000-0001-7365-9576} \and
Kai Jin\inst{1}\orcidID{0000-0003-3720-5117} \and
Ruixi Luo\inst{1}\orcidID{0000-0003-0483-0119} \and 
Song Cao\inst{1}\orcidID{0009-0002-1760-3820}} 

\authorrunning{Taikun Zhu et al.}
\institute{School of Intelligent Systems Engineering, Shenzhen Campus of Sun Yat-sen University, Shenzhen 518107, China
\email{\{zhutk3,luorx,caos6\}@mail2.sysu.edu.cn,jink8@mail.sysu.edu.cn}}

\maketitle              % typeset the header of the contribution
\begin{abstract}
We investigate optimal social welfare allocations of $m$ items to $n$ agents with binary additive or submodular valuations. For binary additive valuations, we prove that the set of optimal allocations coincides with the set of so-called \emph{stable allocations}, as long as the employed criterion for evaluating social welfare is strongly Pigou-Dalton (SPD) and symmetric. Many common criteria are SPD and symmetric, such as Nash social welfare, leximax, leximin, Gini index, entropy, and envy sum.  We also design efficient algorithms for finding a stable allocation, including an $O(m^2n)$ time algorithm for the case of indivisible items, and an $O(m^2n^5)$ time one for the case of divisible items. The first is faster than the existing algorithms or has a simpler analysis. The latter is the first combinatorial algorithm for that problem. It utilizes a hidden layer partition of items and agents admitted by all stable allocations, and cleverly reduces the case of divisible items to the case of indivisible items.In addition, we show that the profiles of different optimal allocations have a small Chebyshev distance, which is 0 for the case of divisible items under binary additive valuations, and is at most 1 for the case of indivisible items under binary submodular valuations.

\end{abstract}

\newcommand{\p}{\mathbf{p}}
\newcommand{\q}{\mathbf{q}}

\section{Introduction}
Maximizing social welfare or minimizing inequality in allocating resources to agents is an important topic in social economics and has been studied extensively in recent years \cite{CPB1989,darmann2015maximizing,ATKINSON2015xvii,barman2018greedy}.
 Each agent has her own subjective \emph{valuation function} over subset of resources (items).  
 Yet how to evaluate the welfare or inequality has no unified answer.
 Some may suggest LexiMin as the criterion \cite{kleinberg1999fairness} -- where we maximize the smallest valuations of all agents, then maximize the second smallest valuations of them, and so on; while others may suggest Maximum Nash Welfare (MNW) \cite{barman2018greedy}, where we maximize the product of valuations of all agents.
 For most classes of valuation functions such as the additive valuation functions (where an agent's valuation is the sum of her valuations for individual items in her bundle), the optimal allocations vary with the selected criteria.

 For the special case of \emph{binary additive valuations} on items, however, Aziz and Rey \cite{10.5555/3491440.3491446} showed that the LexiMin allocations are equivalent to the MNW allocations. 
 This result raises a natural question that whether there are more connections among other criteria. 
 Benabbou et al. \cite{10.1145/3485006} gave a positive answer --  the optimal allocations under any \emph{strongly Pigou-Dalton (SPD)} principle (a.k.a. transfer principle) \cite{moulin1991axioms} over all allocations with maximal utilitarian social welfare (USW, defined as the sum of valuations, a promise of efficiency) are consistent with LexiMin allocations, even under \emph{binary submodular valuations} which subsumes binary additive valuations.
 Roughly, the Pigou-Dalton principle requires that if some income is transferred from a rich person to a poor person, while does not bring the rich to a poorer situation than the poor, then the measured inequality should not increase (or decrease in the strong version); see its formal definition in \Cref{sect:prel}. The SPD principle is a principle admitted by most common criteria.

 %% Clearly, the profiles of (Leximin) optimal allocations (Leximin can be replaced to other SPD criteria according to the results mentioned above) must be equivalent up to permutation, where a profile is the valuations vectors of $n$ agents. Yet there may be many different profiles -- e.g., the number of items allocated to agent $i$ may differ in different optimal allocations and one may wonder how the value of agents can differ among all the optimal profiles. 

 Inspired by the aforementioned consistency on optimums among all the SPD criteria \cite{10.1145/3485006},
   we conduct a further study of the optimums under the SPD criteria
     mainly for the following three scenarios: 
 
  \begin{description}
   \item[IND.] Indivisible items and agents with binary additive valuations. \smallskip
    \item[DIV.] Divisible items and agents with binary additive valuations. \smallskip
   \item[IND-SUB.] Indivisible items and agents with binary submodular valuations. \smallskip
 \end{description}

 % % A \emph{constrained allocation} refers to an allocation satisfying the following constraints:
 % %   the $i$-th agent is allocated no more than $b_i$ and no less than $a_i$ items, for $i\in [n]$.
 % %   Formally, $a_i\leq h_i\leq b_i$ for $i\in [n]$, where the bounds $a_i$'s and $b_i$'s are given.

 Our main results are summarized below.
 We only concern the allocations with maximal USW to ensure the efficiency like \cite{10.1145/3485006}.  
 Unless otherwise stated, the criteria for evaluating inequality is SPD and symmetric.
 Moreover, a profile of an allocation $\chi$ refers to the valuations vector of the $n$ agents under $\chi$.
 \begin{enumerate} [label=(\alph*)]
     \item For IND, the profiles of optimal allocations have Chebyshev distance at most $1$.   
       Moreover, there is a \emph{layer structure} hidden behind the optimal allocations --
       the agents and items are partitioned into serval layers so that 
          items can only be allocated to the agents within the same layer in all optimal allocations. \smallskip
      
     \item In DIV, the profiles of optimal allocations are all the same (in other words, the Chebyshev distance is $0$).
         Halpern et al. \cite{halpern2020fair} showed a similar result in which only LexiMin and MNW are considered.
        
       A layer structure still exists (and with more layers compared to the scenario of IND).
       More importantly, by utilizing this layer structure and using a reduction to IND, we design the first combinatorial algorithm for finding the optimal allocation for DIV scenario, which runs in $O(m^2n^5)$ time. \smallskip

     \item In IND-SUB, the profiles of optimal allocations have Chebyshev distance at most $1$.
       This extends the corresponding result in IND.
       However, the layer structure mentioned above does not hold under this scenario.
 \end{enumerate}

 \begin{table}[ht]
     \centering
     \begin{tabular}{lllcc}
         \toprule
         Scenario  & \multirow{2}{2cm}{\centering SPD Criteria \\ Consistency} 
         & \multirow{2}{1cm}{\centering Running \\Time}
         & \multirow{2}{1cm}{\centering Cheby. \\ Dist.} 
         & \multirow{2}{1.5cm}{\centering Layer \\ Structure}\\
         &  &   \\    \midrule
         IND     & Yes   & $O(m^2n)$& $\leq 1$   & Yes \\
         %Indivisible+    & Yes   & $O(m^2n)$ & $\leq 1$   & No \\
         DIV       & Yes   &$O(m^2n^5)$ & 0    & Yes \\
         %Divisible+      & Yes   & poly      & 0      & No\\
         IND-SUB & Yes\cite{10.1145/3485006} & poly\cite{babaioff2021fair} & $\leq 1$  & No\\
         \bottomrule
     \end{tabular}
     \caption{A summary of the results.}\label{tab:summary}
 \end{table}

 See \Cref{tab:summary} for a summary of the results.
 In all cases, we derive an ``almost consistency'' among different optimums (SPD criteria). 
 It states that the valuation of any agent differs by at most $1$ (or $0$ for DIV case), 
   under any two different optimal allocations. 

 \paragraph{Related work.}

 Halpern et al. \cite{halpern2020fair} show that under binary additive valuations, given any \emph{fractional MNW allocation} (i.e. the MNW allocations in the scenario of DIV), one can compute, in polynomial time, a \emph{randomized allocation} with only \emph{deterministic MNW allocations} (i.e. the MNW allocations in the scenario of IND) in its support and the randomized allocation implements the given fractional MNW allocation. 
 This is a compelling connection between the deterministic and fractional MNW allocations: given a fractional MNW allocation, one can find a convex combination of deterministic MNW allocation to yield it.
 We note that the connection found by them is not a computational method for fractional MNW allocation (since they need a fractional MNW allocation as input) and our method finds an optimal allocation of the scenario of DIV with algorithm designed for computing optimal allocations of the scenario of IND.

 For computational tractability, the SPD optimal allocations can be computed in polynomial time (by computing a Leximin or MNW allocation) in the scenario of IND \cite{kleinberg1999fairness,darmann2015maximizing,barman2018greedy}, DIV \cite{nace2002polynomial} and IND-SUB \cite{10.5555/3491440.3491446}. 
 % and we extended the results to the scenario of Indivisible+ and Divisible+.
 We propose a method to find an optimal allocation of the scenario of DIV 
   with algorithm designed for computing optimal allocations of the scenario of IND. 
 This is reminiscent of the relation between integer programming and linear programming. 
 The well-known branch-and-bound method uses linear programming as a subprogram to solve integer programming problem.
 In this paper, based on nontrivial observation, we split each item into a fixed number of pieces and prove that the optimal allocation over the pieces (viewed as indivisible items) is exactly an optimal allocation of the scenario of DIV.

 The setting of binary valuation is considered in the resource allocation problem, optimal jobs scheduling, load balancing problem. 
 Lin and Li \cite{lin2004parallel} study the special case in which each job can be processed on a subset of allowed machines and its run-time in each of these machines is 1 and find the minimum makespan in polynomial time. 
 Kleinberg et al. \cite{kleinberg1999fairness} study case called uniform load balancing which is to assign jobs to machines so that the set of allocated bandwidth is Leximin. 
 The object equals to find the allocation that the number of jobs assigned to machine is Leximax optimal: lexicographically minimizing the number of jobs assigned to machine when sorted from large to small. 

 Another classic class of resource allocation problems is that with only one kind of resource. 
 In this setting, we only care about the number of items allocated to each agent, rather than the specific subset of items. 
 Ibaraki and Katoh\cite{ibaraki1988resource} make a review on this class of resource allocation problem, viewing as an optimization problem. 
 Allocate a fixed amount of resources (continuous or discrete) to $n$ agents for optimizing the objective function (e.g. separable, convex, minimax, or general). 
 In particular, an important case is the discrete resource allocation problem with a separable convex objective function, where Michaeli and Pollatschek\cite{Michaeli1977OnSN} discusse some properties between the optimal solution of the discrete version and the one of continuous version. These properties can be used to speed up the search for integer solution.

 The resource allocation problems under ternary valuations, which is a natural extension of binary valuations, are much harder to handle. Under additive valuations, Golovin \cite{golovin2005max} proves that it is NP-hard to compute a $(2-\epsilon)$ approximate maximin allocation even with agents' valuations of $\{0,1,2\}$ on single items.
 Another extension of binary valuations is the case where item $j$ has value $p_j$ or $0$ for each agent. In this case of valuations, for computing maximin allocation, Bezakova and Dani\cite{Dani2005santa} pove that there is no approximation alogrithm with performance guarantee better than $2$ unless $P=NP$ and Bansal and Sviridenko\cite{bansal2006santa} present an $O(\frac{\log{\log{m}}}{\log{\log{\log{m}}}})$ approximaition algorithm.

\section{Preliminaries} \label{sect:prel}
For an integer $k>0$, let $[k]$ denote $\{1,...,k\}$. 
Throughout this paper, $[m]$ refers to the set of $m$ items and $[n]$ refers to the set of $n$ agents.

Each agent $i \in [n]$ has a valuation function $v_i: 2^{[m]} \rightarrow{\mathbb{R_+}}$ over subsets of $[m]$ (called bundles) where $v_i(\emptyset) = 0$.
Given a valuation function $v_i$, we define the \emph{marginal gain} of an item $o$ over a bundle $S \in [m]$ as $\Delta_i(S;o) = v_i(S \cup {o}) - v_i(S)$.
We focus on the binary valuations where the marginal gain $\Delta_i(S;o) \in \{0,1\}$. 

Two kinds of binary valuations are discussed frequently in literature, which we call \textbf{0/1-add} and \textbf{0/1-sub}.
For the 0/1-add valuations, the value of a set of items for an agent is the sum of the valuation of the individual items;
the marginal gain $\Delta_i(S;o)$ is based on whether agent $i$ likes item $o$ (and independent of $S$). 
For the 0/1-sub valuations, the marginal gain $\Delta_i(S;o)$ does not increase when $S$ grows; formally, $\Delta_i(T;o) \leq \Delta_i(S;o)$ for $S \subset T \subset [m]$ and $o\in[m]\setminus T$.

Note that the 0/1-sub valuations subsume the 0/1-add ones.

\smallskip
An allocation $\chi$ refers to a collection of disjoint bundles $\chi_1 \dots \chi_n$ such that $\chi_1 \cup \dots \cup \chi_n \subseteq [m]$. 
An allocation $\chi$ is \emph{clean} if all the bundles are clean -- $\chi_i$ is clean if it has no items with zero marginal gain (i.e. for all $o\in\chi_i$, $\Delta_i(\chi_i\setminus\{o\};o)=1$).
For 0/1-sub valuations, $\chi_i$ is clean if and only if $v_i(\chi_i) = |\chi_i|$ (Proposition 3.3 of \cite{10.1145/3485006}).

%% A clean allocation may not allocate all items (for example, two items $o_1$ and $o_2$, and one agent with binary submodular valuation: $v_1(\{o_1\})=v_1(\{o_2\})=v_1(\{o_1,o_2\})=1$).

Given a clean allocation $\chi$, assuming agent $i$ gets $h_i=|\chi_i|$ items under $\chi$,
we call vector $(h_1,\dots,h_n)$ the \emph{profile} of $\chi$, denoted by $\p(\chi)$.
Henceforth, $h_i$ always refers to $|\chi_i|$ unless otherwise stated.

\begin{definition}[Criterion] 
A \emph{criterion of income inequality} (\emph{criterion} for short), a.k.a. \emph{income inequality metric} \cite{ATKINSON2015xvii,CPB1989}, is a function from the profiles to $\mathbf{R}$: Each profile is evaluated by a real number (called \emph{score}); the lower the score, the better the profile under this criterion. 
Following the convention, a criterion must be \emph{symmetric}; i.e.,
it should evaluate all permutations of $\p$ equally.
\end{definition}

\newcommand{\NSW}{\mathsf{NSW}^\neg}
\newcommand{\ES}{\mathsf{EnvySum}}
\newcommand{\GINI}{\mathsf{GiniIndex}}
\newcommand{\entropy}{\mathsf{Entropy}^\neg}
\newcommand{\Congestion}{\mathsf{Congestion}}
\newcommand{\Variance}{\mathsf{Variance}}

\newcommand{\LexiMin}{\mathsf{LexiMin}}
\newcommand{\LexiMax}{\mathsf{LexiMax}}
\newcommand{\Potential}{\mathsf{Potential}}

\begin{example}[Some commonly used criteria]\label{example:criteria}
Let $h^\uparrow_1,\ldots,h^\uparrow_n$ be the permutation of $h_1,\ldots,h_n$, sorted in increasing order.
Take $\Phi(x)$ to be any strictly convex function of $x$. For example, $\Phi(x)=x^2$.
For every profile $\p=(h_1,\ldots,h_n)$, define 
\footnote{For $\NSW$, we first need to maximize the number of agents with nonzero valuation and then maximize the product of the nonzero valuations.}
\begin{flalign*}
\NSW(\p)      &:= -\prod_{h_i>0}{h_i}; &\Potential_\Phi(\p) &:= \sum_{i=1}^{n}\Phi(h_i) \\
\GINI(\p)     &:= \sum_ i i \cdot h^\uparrow _i; &\ES(\p)  &:={\sum}_{h_i<h_j}(h_j-h_i);\\
%  \Variance(\p)   &:=& \sum_i (h_i-\frac{m}{n})^2;\\
\Congestion(\p)      &:= \sum_i \binom{h_i}{2};&\entropy(\p)  &:= \sum_i \frac{h_i}{m} \log (\frac{h_i}{m}). \\
\LexiMax(\p)    &:= \sum_i m^{h_i};  &\LexiMin(\p)    &:= \sum_i m^{m-h_i}; 
\end{flalign*}
\end{example}

% % The \emph{potential} of a profile $\mathbf{h}=(h_1,\ldots,h_n)$ is defined as
% %     \begin{equation*}
% %         \Potential_\Phi(\mathbf{p}) = \sum_{i=1}^{n}\Phi(h_i). 
% %     \end{equation*}
% % where $\Phi(x)$ is a \textbf{strictly} convex function of $x$. For example, $\Phi(x)=x^2$.

% % \phi_f(h)     &:=& \sum_i f(h_i)\quad (\text{$f$ is convex}).\\
\begin{remark}
By the last two definitions, 
minimizing $\LexiMax$ is equivalent to minimizing the largest valuation of all agents, then minimizing the second largest valuation of them, and so on;
minimizing $\LexiMin$ is equivalent to maximizing the smallest valuation of all agents, then maximizing the second smallest valuation of them, and so on.
\end{remark}

We abbreviate $f(\p(\chi))$ as $f(\chi)$ for any criterion $f$.

\begin{definition}[Strongly Pigou-Dalton]
Profile $\q=(q_1,\ldots,q_n)$ is regarded \emph{more balanced} than $\p=(p_1,\ldots,p_n)$,
if there are $j,k\in [n]$ such that $p_j<p_k$ and both $q_j,q_k$ lie in $(p_j,p_k)$ and $q_i=p_i$ for $i\in [n]\setminus \{j,k\}$
(namely, the incomes of two agents are more balanced in $\q$ whereas all other incomes remain unchanged).
A criterion $f$ is \emph{strongly Pigou-Dalton} (SPD) \cite{moulin1991axioms} 
if $f(\q)<f(\p)$ whenever $\q$ is more balanced than $\p$.
(SPD principle is also known as \emph{transfer principle}.)
\end{definition}

All criteria shown in \Cref{example:criteria} are SPD (proved in \Cref{sect:spd-verify}).

\smallskip
We only consider the allocation with maximal utilitarian social welfare (max-USW, maximizing the sum of the valuations of all agents), otherwise one may minimize $\LexiMax(\chi)$ ($\GINI(\chi)$ etc.) by not allocating any items which is uninteresting. 

Henceforth, unless otherwise stated, allocations are assumed to be \textbf{max-USW} and \textbf{clean} (to this end, we can drop items with zero marginal gain without changing the profile). 

\newcommand{\SA}{\mathcal{S}}

\section{Indivisible items and agents with 0/1-add valuations}

\begin{definition}[Stable allocations in IND scenario]\label{def:stable}
Take an allocation $\chi$ of indivisible items.
For each item $o$ allocated to agent $i$ that can be reallocated to another agent $i'$ (i.e., $v_{i'}(\{o\})=1$),
  build an edge $(i,i')$. 
Moreover, if there is a simple path $(i_1,\ldots,i_k)$ ($k\geq 2$) along such edges,
  we state that $\chi$ admits a \emph{transfer} from $i_1=u$ to $i_k=v$, denoted by $u\rightarrow v$,
  which consists of $k-1$ reallocations along the path, after which agent $u$ loses and $v$ gains one.

A \emph{narrowing transfer} refers to a transfer $u\rightarrow v$ with $h_u\geq h_v+2$.
A \emph{widening transfer} refers to a transfer $u\rightarrow v$ with $h_u\leq h_v$.
Other transfers (i.e. $u\rightarrow v$ with $h_u=h_v+1$) are called \emph{swapping transfers}.

Allocation $\chi$ is called \emph{nonstable} if it admits a narrowing transfer, and
  is called \emph{stable} otherwise.

Denote by $\SA$ the set of stable allocations.
% A transfer $u\rightarrow v$ is \emph{legal} if $h_u-1\geq a_u$ and $h_v+1\leq b_v$.
% In other words, legal transfers can be applied without violating the constraints.
% A constrained allocation is \emph{L-nonstable} if it admits a legal narrowing transfer, and is \emph{L-stable} otherwise.
\end{definition}

\begin{lemma}\label{lemma:crucial-lexmin}
Stable allocations are optimal under $\LexiMin$.
\footnote{In fact, Barman et al. \cite{barman2018greedy} proved that if an allocation is not optimal under $\NSW$, then it admits a narrowing transfer (hence is nonstable). Their result can be easily generalized to any SPD criterion, including $\LexiMin$ (as stated in \Cref{lemma:crucial-lexmin}).}

(As a corollary, their profiles are equivalent under permutation.)
\end{lemma}

 \begin{proof}
 Assume $\chi$ is non-optimal under $\LexiMin$. We shall prove that $\chi$ is non-stable, i.e.,
 it admits a narrowing transfer.

 First, take an allocation $\chi^*$ that is optimal under $\LexiMin$.

 We build a graph $G$ with $n$ vertices.
  If an item is allocated to $j$ in $\chi$ and allocated to $k$ in $\chi^*$, where $k\neq j$,
    build an arc from $j$ to $k$. Note that $G$ may have duplicate arcs.
 Clearly, an arc represents a reallocation (of one item) on $\chi$,
   and $\chi$ becomes $\chi^*$ after all the arcs (i.e. reallocations) are applied.
   
 We decompose $G$ into several cycles $C_1,\ldots,C_a$ and paths $P_1,\ldots,P_b$.
 Denote by $s_i,t_i$ the starting and ending vertices of $P_i$, respectively.
 We assume that $t_j\neq s_k$ for any $j\neq k$; otherwise we connect the two paths $P_j,P_k$ into one path.

 For $0\leq i\leq b$, let $\chi^{(i)}$ be the allocation copied from $\chi$ but applied all the arcs (reallocations) in $P_1,\ldots,P_i$.
 $\chi^{(0)}=\chi$.

 Note that $\chi^{(b)}$ becomes $\chi^*$ after applying the arcs in $C_1,\ldots,C_a$.
 We obtain that $\LexiMin(\chi^{(b)})=\LexiMin(\chi^*)$.
 Further since that $\LexiMin(\chi^*)<\LexiMin(\chi)$, there exists $i~(1\leq i\leq b)$ such that $\LexiMin(\chi^{(i)})<\LexiMin(\chi^{(i-1)})$.

 It follows that in $\chi^{(i-1)}$, we have $h_s\geq h_t+2$ (where $s,t$ denote $s_i,t_i$ respectively, for short).
 It further follows that in $\chi^{(0)}=\chi$, we also have $h_s\geq h_t+2$,
 as $h_s$ never increases and $h_t$ never decreases in the sequence $\chi^{(0)},\ldots,\chi^{(i-1)}$.
 Consequently, $\chi$ admits a narrowing transfer (from $s$ to $t$).
 \end{proof}

\begin{theorem}\label{thm:main}
1. For any SPD criterion, the optimums are exactly $\SA$.

2. We can find a stable allocation in $O(m^2n)$ time.
\end{theorem}

\begin{proof}
1. Fix any SPD criteria $f$, we need to show that
\begin{enumerate}
\item[(1)] A nonstable allocation $\chi$ is non-optimal under $f$.
\item[(2)] A stable allocation $\chi$ is optimal under $f$.
\end{enumerate}
Together, the optimal allocations under $f$ are the stable ones.

\medskip Proof of (1): If $\chi$ is nonstable, it admits a narrowing transfer; 
  denoted by $\chi'$ the allocation after this transfer. 
  Clearly, $\p(\chi')$ is more balanced than $\p(\chi)$, 
    and therefore $f(\chi')<f(\chi)$ due to the assumption that $f$ is strongly Pigou-Dalton.
  
\smallskip Proof of (2): 
Assume $\chi,\chi'$ are stable. By \Cref{lemma:crucial-lexmin}, $\p(\chi)$ is equivalent to $\p(\chi')$ up to permutation,
 therefore $f(\chi)=f(\chi')$ as $f$ is symmetric (remind that we always assume so). So, stable allocations have the same score under $f$.
Further since that nonstable allocations are non-optimal under $f$ (Claim~1),
  all stable allocations admit the same lowest score under $f$ (a.k.a. optimal).

\medskip \noindent
2. Finding a stable allocation reduces to finding the allocation with minimum $\Congestion$ (by Claim~1 of this theorem), which can be found using network flows. \footnote{The algorithm is almost the same as that of Kleinberg et al. \cite{kleinberg1999fairness} for finding a $\LexiMax$ optimal allocation but our analysis is simpler.} Specifically, it reduces to computing the minimum-cost flow in the following network (see \Cref{pic:m^2n}):
  
There are $m+n+2$ nodes, including
    a source node $s$, a sink node $t$, and $m$ nodes $u_1,\ldots,u_m$ representing the items
      and $n$ nodes $v_1,\ldots,v_n$ representing the agents. 
    And there are $\Theta(mn)$ edges in the network: 
(i) an edge from $s$ to each $u_i$, with capacity 1 and cost 0;
(ii) an edge from $u_i$ to $v_j$ if agent $j$ likes item $i$, with capacity 1 and cost 0; 
(iii) $m$ edges from each $v_j$ to $t$, in which the $k$-th one has capacity 1 and cost $k-1$.

\begin{figure}[h]
\centering \includegraphics[width=7.8cm]{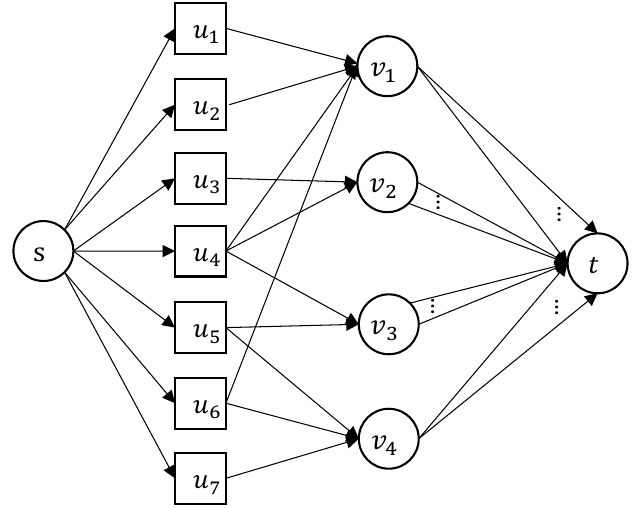}
\caption{Reduction to minimum-cost flow problem.} \label{pic:m^2n}
\end{figure}

Our target -- a flow of size $m$ with the minimum cost --
 can be computed by the Successive Shortest Path algorithm\cite{EK-MCF},
 which increases the size of the current flow by 1
   via augmenting along the shortest path in the residual graph, repeating $m$ times.
For our particular network, finding such a path reduces to 
     finding a non-used edge $(v_j,t)$ with lowest cost 
        such that $s$ can reach $v_j$ in the residual network,
     which can be done in $O(mn)$ time by BFS. In total it runs in $O(m^2n)$ time.
\end{proof}
    
\begin{remark}
It might be asked whether a better allocation under some criterion 
      is also a better allocation under another criterion?
  This answer is no, although the best remains the best crossing different criteria (\Cref{thm:main}).
  For example, for $m=14$ and $n=3$, we have
  
    $\Congestion((0,5,9))=46<47=\Congestion((2,2,10))$, and
    
    $\ES((0,5,9))=18>16=\ES((2,2,10))$.
\end{remark}

%\cite{EK-MCF, MF&MCF-survey}. 

\subsection{Layer partition of agents and items}\label{subsect:layer partition}

Recall that the profiles of stable allocations are equivalent up to permutation (see the corollary part in \Cref{lemma:crucial-lexmin}).
Yet the profiles are not unique -- e.g., the number of items $h_i$ allocated to agent $i$ may differ in different stable allocations.
It raises a natural question that to what extent can $\p(\chi)$ differ for different $\chi$ in $\SA$?
  
Our next theorem shows that the difference is negligible. Recall that the Chebyshev distance of two vectors is the maximum difference along any coordinate dimension.

\begin{theorem} \label{thm:cheb}
    For $\chi,\chi'\in \SA$, it holds that $|h_i-h'_i|\leq 1$ for each agent $i\in [n]$, where $(h_1,\ldots,h_n)=\p(\chi)$ and $(h'_1,\ldots,h'_n)=\p(\chi')$.
In other words, the Chebyshev distance $D(\p(\chi),\p(\chi'))\leq 1$.
\end{theorem}

 \begin{proof}
 Suppose to the opposite that $h_i\geq h'_i+2$. We shall prove that $\chi$ admits a narrowing transfer and hence is nonstable, which contradicts the assumption $\chi\in \SA$.

 We build a graph $G$ with $n$ vertices.
  If an item is allocated to $j$ in $\chi$ and allocated to $k$ in $\chi'$, where $k\neq j$,
    build an arc from $j$ to $k$. Note that $G$ may have duplicate arcs.
 Clearly, $\chi$ becomes $\chi'$ after all the arcs (i.e. reallocations) are applied.
   
 Decompose $G$ into paths $P_1,\ldots,P_b$ and a few cycles.
 Denote by $s_j,t_j$ the starting and ending vertices of $P_j$, respectively.
 We assume that $t_j\neq s_k$ for $j\neq k$ as in \Cref{lemma:crucial-lexmin}.
 Moreover, assume that $s_1=s_2=i$ (due to $h_i\geq h'_i+2$).

 For $0\leq j\leq b$, let $\chi^{(j)}$ be the allocation copied from $\chi$ but applied all the arcs (reallocations) in $P_1,\ldots,P_j$.
 $\chi^{(0)}=\chi$.

 Because $\p(\chi^{(b)})=\p(\chi')$, we have $\LexiMin(\chi^{b})=\LexiMin(\chi')$.
 We also know $\LexiMin(\chi^{(0)})=\LexiMin(\chi')$ since both $\chi^{(0)},\chi'$ are stable (\Cref{lemma:crucial-lexmin}). Together,
 $\LexiMin(\chi^{(0)})=\LexiMin(\chi^{(b)})$. It further implies that 
 $\LexiMin(\chi^{(0)})=\LexiMin(\chi^{(1)})= \ldots = \LexiMin(\chi^{b})$.
 Otherwise, there exists $j$ such that $\LexiMin(\chi^{(j)})$ < $\LexiMin(\chi^{(j-1)})$,
  which means $\chi^{(j-1)}$ admits a narrowing transfer $s_j\rightarrow t_j$,
 implying that $\chi^{(0)}$ admits a narrowing transfer $s_j\rightarrow t_j$.

 Since $\LexiMin(\chi^{(0)})=\LexiMin(\chi^{(1)})=\LexiMin(\chi^{(2)})$,
   we know $s_1\rightarrow t_1$, i.e. $i\rightarrow t_1$ (recall $s_1=i$), is a swapping transfer of $\chi^{(0)}$,
   whereas $s_2\rightarrow t_2$, i.e. $i\rightarrow t_2$ (recall $s_2=i$), is a swapping transfer of $\chi^{(1)}$.
   It follows that $h_{t_2}=h_{i}-2$,
     therefore $\chi^{(0)}=\chi$ admits a narrowing transfer $i\rightarrow t_2$. 
 \end{proof}

\begin{remark}
We will see at later part of this paper (\Cref{thm:01SUB-Chev}) that for the scenario of IND-SUB, the Chebyshev distance between different optimal allocations remains to be at most 1.
The proof of this generalized result is much more involved and is given in \Cref{sec:IND-SUB}.
\end{remark}

According to Theorem~\ref{thm:cheb},
  the income $h_i$ under different stable allocations $\chi$ does not differ too much for every agent $i$,
  so every solution $\chi$ seems to be acceptable for all of them.
However, many questions regarding stable allocations remain to be settled. For example:
\begin{itemize}
\item[$Q_1$.] Are there common properties of all stable allocations?
\item[$Q_2$.] Can we obtain the range of $h_i$ over all stable allocations? 
\item[$Q_3$.] Is it possible to count efficiently the number of stable allocations?
\item[$Q_4$.] Can we find out the profile (of some stable allocation) that optimizes a specific function of $h_1,\ldots,h_n$?
\end{itemize}
 
In what follows, we introduce a construct, called ``\emph{layer partition}'' of agents and items, with some nice property (see \Cref{lemma:layer-partition} below). This construct elucidates the structure of $\SA$ and, as we will demonstrate, provides a pathway to answering the four questions above. 
 
\newcommand{\layer}{\mathsf{layer}}
\newcommand{\Layer}{\mathsf{Layer}}

\medskip First of all, we point out that each agent $i$ falls into two cases:
(1) its income $h_i$ is always equal to some constant $d$ for all $\chi\in \SA$;
or (2) its income $h_i$ equals $d$ for some $\chi\in \SA$ 
  and equals $d-1$ for some other $\chi\in \SA$, for some integer $d>0$.
  This is a corollary of \Cref{thm:cheb}.
Denote by $\layer_d$ the set of agents $i$ for which $h_i$ always equal to $d$ (for $\chi\in \SA$),
  and by $\layer^-_d$ the set of agents $i$ for which $h_i$ sometimes equal $d$ and sometimes equal $d-1$ (for $\chi\in \SA$). 

\begin{lemma}\label{lemma:layer-partition}
    The set of items allocated to $\layer_d$ ($d\geq 0$) is invariant under every stable allocation (denote it by $\Layer_d$ henceforth),
    and the set of items allocated to $\layer^-_d$ (for $d>0$) is invariant as well
       (denote it by $\Layer^-_d$ henceforth).
 Equivalently speaking, corresponding to the partition of agent $\layer_0$, $\layer_1$, $\layer^-_1$, $\layer^-_2$, $\ldots$, there exists a partition of items $\Layer_0$, $\Layer_1$, $\Layer^-_1$, $\Layer^-_2$, $\ldots$, so that under each stable allocation $\chi\in \SA$, the items from $\Layer_d$ (or $\Layer^-_d$, respectively) are always allocated to the agents from $\layer_d$ (or $\layer^-_d$, respectively).
\end{lemma}

Shortly, items in any given layer can only be allocated to agents in that layer.
To be clear, we regard $\layer_d,\Layer_d$ as in the same layer and $\layer^-_d,\Layer^-_d$ as in another layer.

\subsection{Proof of \Cref{lemma:layer-partition} and the explicit construction of the layer partition}

We prove \Cref{lemma:layer-partition} and show how to compute
$\layer_d,\layer^-_d,\Layer_d,\Layer^-_d$ efficiently in this subsection. 
To this end, we introduce some additional notation first. 

Fix any stable allocation $\chi\in \SA$. 
Recall the \emph{transfers} on $\chi$ (and the notation $\rightarrow$) in \Cref{def:stable}.
Those agents receiving $d$ items under $\chi$ can be classified into three groups:
\begin{eqnarray}
 p_d &:=& \{i \mid h_i=d \text{ and } \exists i\rightarrow j\text{ with }h_j=d-1\}; \label{eqn:def_pd}\\
 q_d &:=& \{i \mid h_i=d \text{ and } \exists k\rightarrow i\text{ with }h_k=d+1\}; \label{eqn:def_qd}\\
 r_d &:=& \{i \mid h_i=d, i\notin p_d, \text{ and }i \notin q_d\}.      \label{eqn:def_rd}
\end{eqnarray}
See \Cref{pic:layers-proof}. Observe that $p_d$ is disjoint with $q_d$. Otherwise, there is $k\rightarrow j$ with $h_k=h_j+2$, contradicting the fact that $\chi$ is stable. Denote $s_d=p_d\cup q_{d-1}$. 
Be aware that the quantity $h_i$, the groups $p_d,q_d,r_d$, and the set $s_d$ all depends on the selected stable allocation $\chi$.

\begin{figure}[h]
  \centering
  \includegraphics[width=.5\columnwidth]{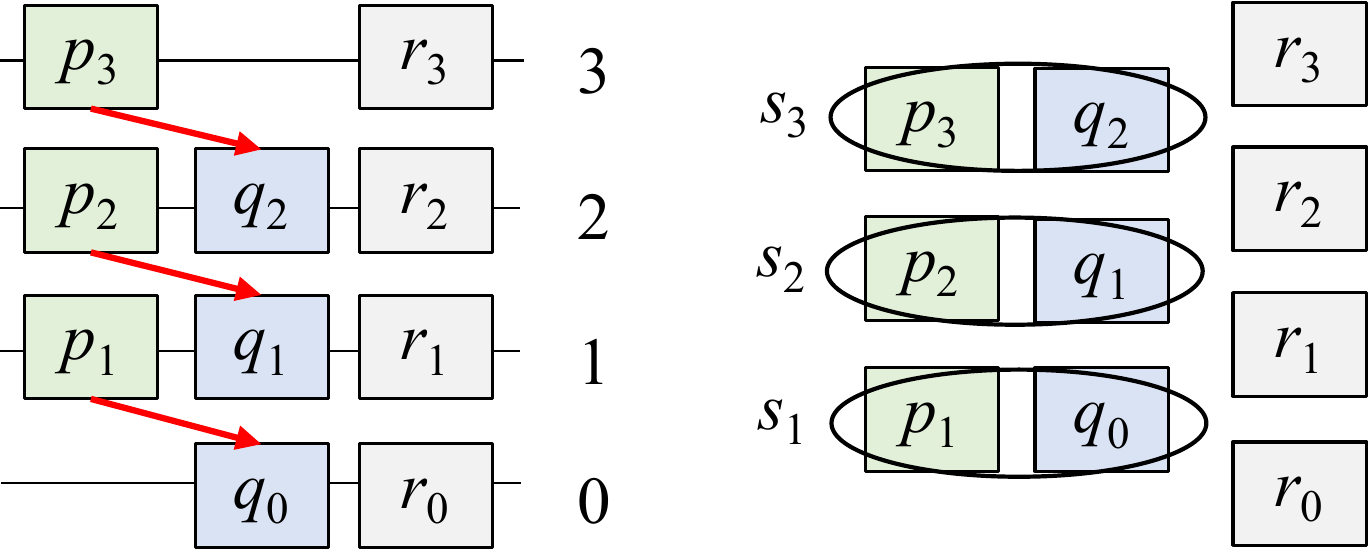}\quad
  \caption{Illustration of $p_d,q_d,r_d,s_d$.}\label{pic:layers-proof}
\end{figure}

Denote by $R_d$ and $S_d$ the set of items allocated to $r_d$ and $s_d$ under $\chi$, respectively.
  
Moreover, within this subsection, we rank all items and agents as follows.
\begin{itemize}
\item elements in $r_0,R_0$ have rank $0$; \quad elements in $s_1,S_1$ have rank $1$; 
\item elements in $r_1,R_1$ have rank $2$; \quad elements in $s_2,S_2$ have rank $3$; and etc.
\end{itemize}

\begin{proposition}\label{prop:layer-agents}
In the aforementioned stable allocation $\chi$,
\begin{enumerate}
\item there is no transfer from $s_d$ to any lower-ranked agents (namely, $r_{d-1},s_{d-1},\ldots$);
\item there is no transfer from $r_d$ to any lower-ranked agents (namely, $s_d,r_{d-1},\ldots$);
\end{enumerate}
\end{proposition}

This easily follows from the definitions of $p_d,q_d,r_d$. Trivial proofs omitted.

In the following, the rank refers to the rank defined above with respect to the fixed $\chi$.

\begin{proposition}\label{prop:layer-items} 
In any stable allocation $\chi'$, it holds that
\begin{enumerate}
\item a higher-ranked item cannot be allocated to a lower-ranked agent.
\item a lower-ranked item cannot be allocated to a higher-ranked agent.
\item Together, items in $R_d$ are always assigned to $r_d$, and items in $S_d$ are always assigned to $s_d$.
\end{enumerate}
\end{proposition}

\begin{proof}
1. An item $o$ from $S_d$ (or $R_d$, respectively) is not attractive to any lower-ranked agent,
  since otherwise $\chi$ admits a transfer from $s_d$ (or $r_d$, respectively) to a lower-ranked agent (using just one reallocation of $o$), which contradicts \Cref{prop:layer-agents}.
  As a result, an item from $S_d$ (or $R_d$, respectively) cannot be allocated to any lower-ranked agent in $\chi'$.

\medskip 2. We prove it by induction.
In this proof, for the simplicity of presentation, assume $\max\{h_i\}=3$; see \Cref{pic:layers-proof}.
First, the items ranked lower than $r_3$ cannot be allocated to $r_3$ in $\chi'$.
Otherwise, as the items in $R_3$ must be allocated to $r_3$ (by the analysis above),
  some agent in $r_3$ receives more than 3 items by the pigeonhole principle,
    which implies that $\p(\chi')$ is non-optimal under $\LexiMax$, and therefore $\chi'$ is not stable (\Cref{thm:main}).
Similarly, the items ranked lower than $s_3$ cannot be allocated to $s_3$ in $\chi'$, otherwise $\p(\chi')$ contains more agents receiving $3$ items than $\p(\chi)$ which is impossible. So on and so forth.
\end{proof}

% Be aware that argument~(1) holds for all allocations whereas argument~(2) holds only for stable allocations.

\begin{proposition}
For $i\in r_d$, it holds that $h_i$ always equal $d$ for all $\chi\in \SA$;
    and for $i\in s_d$, it holds that $h_i$ can be $d$ or $d-1$ under different $\chi\in \SA$.
It follows that $\layer_d=r_d$ and $\layer^-_d=s_d$. 
\end{proposition}

\begin{proof}
    According to \Cref{prop:layer-items} (claim 3), in every stable allocation,
     agents $r_d$ will receive and only receive $R_d$. Note that $|R_d| = d |r_d|$.
   Therefore, in every stable allocation,
     each agent from $r_d$ receives exactly $d$ items (otherwise, there must be one agent receiving more than $d$ and one agent receiving less than $d$,
        which is worse than equally distribution for the criteria $\LexiMin$ and hence nonstable).
     Hence, $h_i$ always equals $d$ for $i\in r_d$. 
             
   For each $i\in p_d$, we know $h_i=d$ and $h_i$ can be reduced to $d-1$ (due to (\ref{eqn:def_pd})) in another stable allocation.
   For each $i\in q_{d-1}$, we know $h_i=d-1$ and $h_i$ can be increased to $d$ (due to (\ref{eqn:def_qd})) in another stable allocation.
   Therefore $h_i$ can be $d-1$ or $d$ for $i\in p_d\cup q_{d-1}=s_d$.
\end{proof}

\Cref{lemma:layer-partition} simply follows from all the analysis above.
In particular, combining the last two propositions, we obtain that $\Layer_d=R_d$ and $\Layer^-_d=S_d$.

\subsection{Some applications of the layer partition -- answering $Q_1,\ldots,Q_4$}\label{sect:layer-answer-Q}

With the layer partition of agents and items, we can promptly answer $Q_1$ to $Q_4$.

\begin{description}
\item[Answer of $Q_1$.] They all allocate items layer by layer; in which the layers are invariant with respect to the chosen allocation.
\smallskip 
\item[Answer of $Q_2$.] For $i\in \layer_d$, the range of $h_i$ is $[d,d]$.
        For $i\in \layer^-_d$, the range of $h_i$ is $[d-1,d]$.
        Computing the layers reduces to computing $r_d,s_d$ for any fixed $\chi$, which is easy.
        Hence we can compute the ranges of $h_i$'s efficiently.
        
        Alternative methods for computing the ranges have to use network flow and are far more complicated.
\smallskip 
\item[Answer of $Q_3$.] Counting stable allocations reduces to counting stable allocations within each layer (and then applying the rule of product).
    
     However, this counting problem is \#P-hard, as the problem of counting stable allocation in $\layer_1$ can be reduced to counting perfect matchings which is already \#P-hard.
\smallskip 
\item[Answer of $Q_4$.] Finding a profile (of a stable allocation) that 
   minimizes a linear function of $h_1,\ldots,h_n$, e.g. $\sum c_ih_i$, is easy.
   Modify the network used in the proof of \Cref{thm:main}~Claim~2 as follows:
      Among the edges from $v_j$ to $t$, change the cost of the $k$-th one to be $(k-1)\times A + c_i$,
       where $A$ is a large enough constant.
   Then, the minimum-cost flow still has the minimum $\Congestion$, and it optimizes $\sum c_ih_i$.
   For non-linear functions, it is more difficult. Yet the layer partition still helps break down the task.
\end{description}

\begin{remark}
    Halpern et al. \cite{halpern2020fair} implicitly found a structure similar to our layer structures. 
Specifically, (in the proof of their Theorem 4) given a fractional MNW allocation, they partition the agents into subsets according to the floor of valuations (so in a certain set, the valuation range of each agent is $[x,x+1)$ for some integer $x$).
They imply that the agents in each set must be fully allocated to a certain subsets of items for any fractional MNW allocation, 
and the partition and correspondence between the subsets of agents and the subsets of items also hold for deterministic MNW allocations. We discover our layer structures independently. 
Our layer structures are more specific, which can be briefly explained through the valuation range of each agent in a certain layer: for the scenario of DIV, compared to their range of $[x,x+1)$ for some integer $x$, the range in our layer structures is some fixed rational number and for the scenario of IND, compared to their range of $[x,x+1]$ for some integer $x$, the range in our layer structures is $[x,x+1]$ for some integer $x$ or some fixed integer (the range of $[x,x]$). 
Additionally, to compute the layer partition for deterministic MNW allocations, a fractional MNW allocation is necessary through their framework, while we can compute it directly in the scenario of IND (from a deterministic MNW allocation in their term).
\end{remark}

\section{Divisible items and agents with 0/1-add valuations}
 This section discusses the scenario of DIV, i.e., the case of divisible items in which we are allowed to allocate a part of an item to an agent
   (and different parts to different agents perhaps).
 Note that the 0/1-sub valuations cannot match easily with divisible items.
  Therefore we restrict ourselves to 0/1-add valuations in this section.

 \begin{definition}[Stable allocations in scenario of DIV]\label{def:stable-divisible}
 Let $\chi$ be an allocation of divisible items.
 For each part of the item $o$ allocated to agent $i$ that can be reallocated to another agent $i'$ (i.e., $v_i(\{o\})=1$),
 build an edge $(i,i')$. 
 Moreover, if there is a simple path $(i_1,\ldots,i_k)$ ($k\geq 2$) along such edges,
 we state that $\chi$ admits a \emph{transfer} from $i_1=u$ to $i_k=v$, denoted by $u\rightarrow v$.
 It consists of $k-1$ reallocations with $\Delta>0$ fraction of items along the path,
   after which $h_u$ decreases by $\Delta$ and $h_v$ increases by $\Delta$.

 A \emph{narrowing transfer} refers to a transfer from $u$ to $v$ with $h_u-\Delta \geq h_v+\Delta$.
 A \emph{widening transfer} refers to a transfer from $u$ to $v$ with $h_u\leq h_v$.
 An Allocation $\chi$ is called \emph{nonstable} if it admits some narrowing transfer
   and is \emph{stable} otherwise.

 Denote by $\SA^*$ the set of stable allocations in DIV scenario. 
 % A transfer $u\rightarrow v$ is \emph{legal} if $h_u-\Delta\geq a_u$ and $h_v+\Delta\leq b_v$.
 % In other words, legal transfers can be applied without violating the constraints.
 % A constrained allocation is \emph{L-nonstable} if it admits a legal narrowing transfer, and is \emph{L-stable} otherwise.
 \end{definition}

 \begin{lemma}\label{lemma:frac-crucial-lexmin}
 Stable allocations are optimal under $\LexiMin$.
 (As a corollary, their profiles are equivalent up to permutation.)
 \end{lemma}

 \begin{proof}
 	Assume $\chi$ is non-optimal under $\LexiMin$. We shall prove that $\chi$ is non-stable, i.e.,
 	it admits a narrowing transfer.
	
 	First, take an allocation $\chi^*$ that is optimal under $\LexiMin$.
	
 We build a graph $G$ with $n$ vertices.
  Be aware that according to $\chi$ and $\chi^*$, 
    each item $i$ can be divided into several pieces $i_1,\ldots,i_p$ (with total size $1$),
    so that each piece is given to a certain agent (denoted by $j$) in $\chi$ and given to a certain agent (denoted by $k$) in $\chi'$.
  If $j\neq k$, build an arc from $j$ to $k$, with weight equal to the size of this piece.
 Clearly, an arc represents a reallocation (of one piece) on $\chi$,
   and $\chi$ becomes $\chi^*$ after all the arcs (i.e. reallocations) are applied.

 We decompose graph $G$ into several cycles $C_1,\ldots,C_a$ and paths $P_1,\ldots,P_b$, where
   the edges in any cycle or path have the same weight, and where $t_j\neq s_k$ for $j\neq k$, where $s_i,t_i$ denote the starting and ending vertices of $P_i$, respectively.
   Such a decomposition exists under appropriate division of items.
	
 	For $0\leq i\leq b$, let $\chi^{(i)}$ be the allocation copied from $\chi$ but applied all the arcs (reallocations) in $P_1,\ldots,P_i$.
 	$\chi^{(0)}=\chi$.
	
 	Be aware that $\chi^{(b)}$ becomes $\chi^*$ after applying the arcs in $C_1,\ldots,C_a$.
 	We obtain that $\LexiMin(\chi^{(b)})=\LexiMin(\chi^*)$.
 	Further since that $\LexiMin(\chi^*)<\LexiMin(\chi)$, there exists $i~(1\leq i\leq b)$ such that $\LexiMin(\chi^{(i)})<\LexiMin(\chi^{(i-1)})$.
	
 	It follows that in $\chi^{(i-1)}$, we have $h_s> h_t$ (where $s,t$ denote $s_i,t_i$ respectively, for short).
 	It further follows that in $\chi^{(0)}=\chi$, we also have $h_s> h_t$,
 	since $h_s$ never increases and $h_t$ never decreases in the sequence $\chi^{(0)},\ldots,\chi^{(i-1)}$.
 	Consequently, $\chi$ admits a narrowing transfer (from $s$ to $t$). 
 \end{proof}

 \begin{theorem}\label{thm:main-frac}
 1. For any SPD criterion, the optimums coincide with $\SA^*$. \\
 2. We can find a stable allocation in polynomial time.
 \end{theorem}

 Just like we prove \Cref{thm:main}~Claim~1 using \Cref{lemma:crucial-lexmin},
  we can prove \Cref{thm:main-frac}~Claim~1 using \Cref{lemma:frac-crucial-lexmin}  (the proof is omitted).

 To find a stable allocation in polynomial time,
  we can use an approach based on linear programming (LP) which is shown in \Cref{subsect:alg-stable-div} (thus we prove \Cref{thm:main-frac}~Claim~2).
 Alternatively, we find a purely combinatorial approach (which is more efficient and more interesting) for finding such an allocation. It utilizes the layer partition with several nontrivial ideas. See details in \Cref{subsect:comb}.

 \begin{theorem} \label{thm:chev-divisible}
     For $\chi,\chi'\in \SA^*$, it holds that $\p(\chi)=\p(\chi')$, namely, the Chebyshev distance $D(\p(\chi),\p(\chi')) = 0$.
 \end{theorem}

 A proof of \Cref{thm:chev-divisible} is related to the LP approach mentioned above and is given in \Cref{subsect:alg-stable-div}. An alternative proof is similar to the proof of \Cref{thm:cheb} and is omitted.
  
 \subsection{Layer partition of agents and items for divisible items
    \& a combinatorial algorithm for finding a stable allocation}\label{subsect:comb}

 In the following, we extend the layer partition given in \Cref{subsect:layer partition} to the divisible case
   and then present the aforementioned combinatorial algorithm.

 According to \Cref{thm:chev-divisible}, profile $\p(\chi)$ is unique for $\chi\in \SA^*$.
 In other words, the income $h_i$ of each agent $i$ is independent of $\chi$, as long as $\chi$ is stable.

 For each real number $d$, denote by $\layer^*_d$ the set of agents that always receive $d$ items no matter in which stable allocation; formally, $\layer^*_d=\{i \mid h_i=d\}$.
 
 \begin{lemma}\label{lemma:layer_integer}
 The set of items allocated to $\layer^*_d$ is invariant for $\chi\in \SA^*$.
 Moreover, this set (denoted by $\Layer^*_d$ henceforth) consists of complete items only.
 \end{lemma}

 \begin{proof}
 Fix $\chi$. Let $L_d$ be the items allocated to $\layer_d$ under $\chi$.
 An item in $L_{d'}$ ($d'>d$) cannot be allocated to $\layer_d$ otherwise there is a narrowing transfer in $\chi$. 
 An item in $L_{d'}$ ($d'<d$) cannot be allocated to $\layer_d$ otherwise the allocation is not $\LexiMax$ optimal (can be proved by induction as in the proof of \Cref{prop:layer-items}).
 Therefore, those items in $L_d$ can only be allocated to $\layer_d$ (even for other stable allocations).
 We thus obtain the first part of this lemma.

 2. In any stable allocation, an item cannot be allocated to different layers.
 Otherwise there is clearly a simple narrowing transfer. 
 \end{proof}

 \Cref{lemma:layer_integer} implies a layer partition of items and agents, where items $\Layer^*_d$ and agents $\layer^*_d$ are in the same layer. 

 \begin{lemma}\label{lemma:reallocation}
 	A stable allocation $\gamma$ for the divisible case can be obtained by optimally reallocating items $\Layer^-_d$ (regarded as divisible items) to the agents $\layer^-_d$ and items $\Layer_d$ (regarded as indivisible items) to $\layer_d$ for each $d$.
 \end{lemma}

 \begin{proof}
 For each $d$, let $\gamma^{-}_d$ be an optimal allocation of items $\Layer^-_d$ (regarded as divisible items) to the agents $\layer^-_d$. Moreover, let $\gamma_d$ be an optimal allocation of items $\Layer_d$ (regarded as indivisible items) to $\layer_d$ so that each agent in $\layer_d$ receives $d$ items.
 Combine $\gamma_d,\gamma^-_d(d\geq 0)$ in all layers to obtain an overall allocation $\gamma$.
 We claim that $\gamma\in \SA^*$. The proof is as follows.

 Those items in higher layer cannot be given to agents in lower layer -- there are no such edges as shown in the proof of \Cref{prop:layer-items} (this also holds for the divisible case).
 Hence there is no narrowing transfer between layers in $\gamma$. 
 Also, there is no narrowing transfer within any layer of $\gamma$ by the construction of $\gamma$. 
 Together, $\gamma$ admits no narrowing transfer and hence is stable. So, $\gamma\in\SA^*$. 
 \end{proof}

 %The layers are partitioned into sub-layers in the divisible case.

 %We now compare the layer partition $\{\layer^*_d,\Layer^*_d: d \text{ real}\}$ to the one for indivisible case given in \Cref{subsect:layer partition} $\{\layer_d,\Layer_d,\layer^-_d,\Layer^-_d: d \text{ integer}\}$. 

 \begin{lemma}\label{lemma:frac-bound} For any $\chi \in \SA^*$, it holds that

 1. $h_i=d$ for each $i\in \layer_d$. 
 
 2. $h_i\in [d-1,d]$ for $i\in\layer^-_d$.
 \end{lemma}

 \begin{proof}
 Denote $(g_1,\ldots,g_n)=\p(\gamma)$ and $(h_1,\ldots,h_n)=\p(\chi)$.

 Combining (1)-(3) below, we immediately obtain the lemma.
 	\begin{itemize}
         \item[(1)] For each $i\in [n]$, it holds that $h_i=g_i$ (apply \Cref{thm:chev-divisible} with $\gamma,\chi\in \SA^*$).
         \item[(2)] For each $i\in \layer_d$, it holds that $g_i=d$ (trivial).
         \item[(3)] For each $i\in \layer^-_d$, it holds that $g_i \in \left[d-1, d \right]$.
             (Proof: Since $\gamma^-_d$ is ($\LexiMin$) optimal,
       $g_i\geq d-1$. Since $\gamma^-_d$ is also ($\LexiMax$) optimal, $g_i\leq d$.)
 	\end{itemize}  
 \end{proof}

 \begin{lemma}\label{lemma:frac-distance}
 	Given $\chi\in \SA^*$, for two different layers $\layer^*_d$ and $\layer^*_{d'}$ of $\chi$, we have $$\left\vert d-d'\right\vert \geq \frac{1}{n^2}.$$
 \end{lemma}

 \begin{proof}
 	By Lemma~\ref{lemma:reallocation}, we obtain that
 	\begin{equation*}
 		d =  \frac{|\Layer^*_d|}{|\layer^*_d|}, \quad d' = \frac{|\Layer^*_{d'}|}{|\layer^*_{d'}|}.
 	\end{equation*}
 	Recall Lemma~\ref{lemma:layer_integer}. $|\Layer^*_d|$ and $|\Layer^*_{d'}|$ are integers. Thus,
 	\begin{equation*}
 		\left\vert d - d'\right\vert 
 		= \left\vert\frac{|\Layer^*_{d}||\layer^*_{d'}|-|\Layer^*_{d'}||\layer^*_{d}|}{|\layer^*_{d}||\layer^*_{d'}|}\right\vert
 		\geq\frac{1}{n^2}.
 	\end{equation*}
 \end{proof}

 The last three lemmas are crucial to the combinatorial algorithm.

 \medskip For divisible case, we divide each item into $2n^2$ identical (unbreakable) pieces with size $\frac{1}{2n^2}$, and solve this indivisible case (with $2n^2m$ indivisible pieces) using the algorithm in \Cref{thm:main} Claim~2. 

 We obtain an allocation with several layers, each layer's index $d$ is a real number -- a multiple of $\frac{1}{2n^2}$. 
 For those agent $j \in \layer_d$, we have $h_j = d$, and 
 for those $i \in \layer^-_d$, we have $h_i \in [d-\frac{1}{2n^2},d]$. 

 According to Lemma~\ref{lemma:frac-bound}, for the divisible case,
   we have $$h_i\in [d-\frac{1}{2n^2},d] \text{ for }i \in \layer^-_d.$$
  
 As a result, $i \in \layer^*_{d'}$, for some $d'\in [d-\frac{1}{2n^2},d]$. 

 By Lemma~\ref{lemma:frac-distance}, there can only be one such $d'$ in the range of $[d-\frac{1}{2n^2},d]$ for $i \in \layer^-_d$. By Lemma~\ref{lemma:reallocation}, we can obtain an optimal allocation for divisible case just by reallocating items within each layer, and we can calculate that $d'=\frac{|\Layer^-_d|}{|\layer^-_d|}.$

 The obtained algorithm deals with $2mn^2$ items and $n$ agents. Therefore, it has a time complexity of $O(m^2n^5)$, which is still better than the linear programming procedure (explained below) with time complexity of $O(m^{3.5}n^{5.5})$ with $O(mn)$ variables and $O(n^2)$ linear programming problems in the worst case.

 %It remains to be investigated whether we can further improve the running time $O(m^2n^5)$.

 \subsection{Find a stable allocation for divisible items using LP and a proof of Theorem~\ref{thm:chev-divisible}}\label{subsect:alg-stable-div}

 This subsection provides an alternative approach based on linear programming for finding a stable allocation for divisible items. This approach is more straightforward but it is not combinatorial.

 \smallskip Applying \Cref{thm:main-frac}~Claim~1, finding a stable allocation
 	reduces to finding an allocation with the minimum $\LexiMin$.
     We claim that 
     \begin{itemize}
     \item[(1)] the latter further reduces to computing the ``fair multi-flow'' \cite{nace2002polynomial} in the network below:
 	\end{itemize}

 \begin{figure}[h]
 \centering \includegraphics[width=5.8cm]{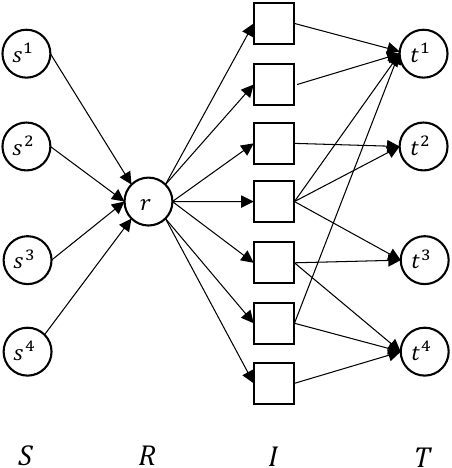}
 \caption{Reduction to fair multi-flow problem.} \label{pic:multi-flow}
 \end{figure}

     \smallskip See \Cref{pic:multi-flow}. There are $n$ source nodes $s^1,\ldots,s^n$ in the first layer $S$. 
 	The second layer $R$ consists of only one node $r$ as a relay. 
 	The third layer $I$ has $m$ nodes, standing for items. 
     The last layer $T$ has $n$ sink nodes $t^1,\ldots,t^n$.
     For each agent $d$, there is an arc $(s^d,r)$ of unlimited capacity. 
     For every item $i$ in $I$, there is an arc $(r,i)$ of capacity 1.  
 	If item $i$ can be allocated to agent $d$, an arc $(i,t^d)$ is added with capacity 1. 

 	In the aforementioned fair multi-flow problem, 
        we shall send a flow $f_d$ from $s^d$ to $t^d$ for each $d$. The sum of the flows on each edge cannot exceed the capacity, and the $n$ amount of flows $|f_1|,\ldots,|f_n|$ as a vector should have the lowest $\LexiMin$.
     Clearly, the solution to this problem
        corresponds to the optimal allocation under $\LexiMin$
          (items are fully allocated as the multi-flow is $\LexiMin$), therefore claim~(1) holds.
	
     \smallskip Finally, recall the result of \cite{nace2002polynomial}
       which states that the fair multi-flow problem can be reduced to a polynomial number of linear programming problems
         and thus be solved in polynomial time. 

\medskip
 \begin{proof}[Proof of \Cref{thm:chev-divisible}]
     Recall finding a stable allocation by computing a fair multi-flow, the flow vector $(|f_1|,\ldots,|f_n|$) is unique for multi flow problem, according to \cite{nace2002polynomial}.
  Therefore the profile of optimal allocation is also unique.
 \end{proof}

\section{Indivisible items and agents with 0/1-sub valuations}\label{sec:IND-SUB}
\begin{sloppypar}

We now move on to the scenario of IND-SUB.

The 0/1-sub function is closely related to the matroid theory \cite{Oxley-10.1093/acprof:oso/9780198566946.001.0001}.
A matroid is a pair $(E,\mathcal{I})$, where $E$ is a finite set (called the \emph{ground} set) and $\mathcal{I}$ is a family of subsets of $E$ (called the \emph{independent} sets). 

The independent sets satisfy the following three axioms:
\begin{enumerate}[label=(I\arabic*)]
	\item $\emptyset \in \mathcal{I}$,
	\item if $Y \in \mathcal{I}$ and $X \subseteq Y$, then $X \in \mathcal{I}$, and
	\item if $X,Y \in \mathcal{I}$ and $|X| < |Y|$, then there exists $y \in Y \setminus X$ such that $X\cup \{y\} \in \mathcal{I}$.
\end{enumerate}

If an agent has a 0/1-sub valuation on items, then the set of \emph{clean} bundles forms the set of independent sets of a matroid (proved in Benabbou et al. \cite{10.1145/3485006}).
Benabbou et al. \cite{10.1145/3485006} further showed that under 0/1-sub valuations, the set of max-USW allocations (under any SPD criterion) are consistent with $\LexiMin$ allocations.
This result generalizes \Cref{thm:main} Claim~1 (for the 0,1-add valuations). 
Besides, Babaioff et al.\cite{babaioff2021fair} showed that under 0/1-sub valuation, the allocation that optimizes $\NSW$ can be found in polynomial time.

\smallskip To the best of our knowledge, however, it is not known 
  to what extent can $\p(\chi),\p(\chi')$ differ for different optimal allocations $\chi,\chi'$.
We answer this question by the following theorem.

\newcommand{\SSA}{\mathcal{SS}}
Let $\SSA$ denote the set of (max-USW and clean) allocations which are optimal under any SPD criterion.

\begin{theorem}\label{thm:01SUB-Chev}
For $\chi,\chi'\in \SSA$, it holds that $|h_i-h'_i|\leq 1$ for each agent $i\in [n]$, where $(h_1,\ldots,h_n)=\p(\chi)$ and $(h'_1,\ldots,h'_n)=\p(\chi')$.
In other words, the Chebyshev distance $D(\p(\chi),\p(\chi'))\leq 1$.
\end{theorem}

\begin{proof} 
Assume (i) $\chi,\chi'\in \SSA$ and (ii) $D(\p(\chi),\p(\chi'))\geq 2$.
If there are multiple such pairs of allocations, 
  take the pair $(\chi,\chi')$ with minimum symmetric difference $\sum_{i \in [n]} | \chi_i \triangle \chi'_i|$.

Without loss of generality, assume that 
    $$h_1\leq h_2 \leq \dots \leq h_n \text{ and }h_q \geq h'_q+2.$$ 

Assume that $h'_{j_1} \leq h'_{j_2} \leq \dots \leq h'_{j_n}$.

Notice that $\p(\chi) \neq \p(\chi')$. Take the minimum index $i$ satisfying $h_i \neq h'_i$. Clearly, for all $k\in[i-1]$ we have
\begin{equation}\label{eq:k}
h_k = h'_k
\end{equation} 
In fact, it must hold that 
\begin{equation}\label{ineq:i}
h_i < h'_i
\end{equation}  
Recall that the smaller $\LexiMin(\p)$, the better the vector $\p$ under $\LexiMin$.
Indeed, if $h_i > h'_i$, then together with \eqref{eq:k}, 
    $$\LexiMin((h_1,\dots,h_i)) > \LexiMin((h'_1,\dots,h'_i)).$$
Moreover, 
    $$\LexiMin((h'_1,\dots,h'_i))\geq \LexiMin((h'_{j_1},\dots,h'_{j_i}))$$ 
     since $h'_{j_1},\dots,h'_{j_i}$ are the smallest $i$ elements in $\p(\chi')$. 
Together, we have $\LexiMin((h_1,\dots,h_i)) > \LexiMin((h'_{j_1},\dots,h'_{j_i}))$.
By comparing the smallest $i$ elements of  $\p(\chi)$ and  $\p(\chi')$, we see $\LexiMin(\p(\chi))>\LexiMin(\p(\chi'))$, contradicting the assumption that $\chi$ and $\chi'$ are both $\LexiMin$ optimal.

By the definition of $i$ and the assumption $h_q \geq h'_q+2$, it is easy to see $q>i$ and hence $h_q\geq h_i$. Further, we claim 
\begin{equation}\label{ineq:q>=i+2}
h_q \geq h_i+2
\end{equation} 
It reduces to prove $h'_q \geq h_i$. 
    Since $\chi$ and $\chi'$ are both $\LexiMin$ optimal, the two multisets $\{h_1,\dots,h_n\}$ and $\{h'_1,\dots,h'_n\}$ are equivalent.
    By \eqref{eq:k}, for all $k\in[i-1], h_k = h'_k$.
Together, the multiset $\{h_i,\dots,h_n\}$ equals $\{h'_i,\dots,h'_n\}$.
Further by the assumption $h_1\leq\dots\leq h_n$, the elements in $\{h'_i,\dots,h'_n\}$ are not smaller than $h_i$.
Recall $q>i$, we have $h'_q \geq h_i$. 

\smallskip
Recall that the family of clean bundles $\mathcal{I}_j=\{S \subseteq [m] \mid v_j(S)=|S|\}$ for $j \in [n]$ forms a family of independent sets of a matroid \cite{10.1145/3485006}. 
By $($I3$)$ of the independent-set matroid axioms and inequality \eqref{ineq:i}, there exists an item $o_1 \in \chi'_i \setminus \chi_i$ making $v_i(\chi_i \cup \{o_1\})=v_i(\chi_i)+1 = h_i+1$. 
And $o_1$ is allocated to some agent $i_1 \neq i$ under allocation $\chi$.
Otherwise, $o_1$ is not allocated to anyone, and can be allocated to agent $i$ violating that $\chi$ is max-USW. Consider the following three cases:
\begin{enumerate}
    \item Suppose $h_{i_1} \geq h_i +2$. Then transferring $o_1$ from $i_1$ to $i$ in $\chi$ decreases $\LexiMin(\p(\chi))$, contradicting that $\chi$ is $\LexiMin$ optimal. 
    
    \item Suppose $h_{i_1} = h_i +1$. We note that $i_1 \neq q$ since $h_q\geq h_i+2$ (inequality \eqref{ineq:q>=i+2}). If we transfer $o_1$ from $i_1$ to $i$ in $\chi$, $\LexiMin(\p(\chi))$ and $h_q$ are unchanged, which means $(\chi,\chi')$ still satisfies the two conditions (i) $\chi,\chi'\in \SSA$ and (ii) $D(\p(\chi),\p(\chi'))\geq 2$,
    but the $\sum_{i \in [n]}|\chi_i \triangle \chi'_i|$ decreases, a contradiction.
    
    \item Suppose $h_{i_1} \leq h_i$. 
    
    We first show $h_{i_1} \leq h'_{i_1}$. This clearly holds by \eqref{eq:k} if $i_1 < i$. 
    
    When $i_1 > i$, since $h_1 \leq \dots \leq h_n$, we have $h_{i_1} \geq h_i$. Together with the assumption $h_{i_1} \leq h_i$ in this case, we have $h_{i_1} = h_i$. 

    Suppose to the oppose that $h_{i_1} > h'_{i_1}$, which means $h_i > h'_{i_1}$. 
    Applying \eqref{eq:k}, 
        $$\qquad \LexiMin((h_1,\dots,h_{i-1},h_i)) > \LexiMin((h'_1,\dots,h'_{i-1},h'_{i_1})).$$ 
    Moreover, 
        $$\LexiMin((h'_1,\dots,h'_{i-1},h'_{i_1})) \geq \LexiMin((h'_{j_1},\dots,h'_{j_i})),$$ 
    since $h'_{j_1},\dots,h'_{j_i}$ are the smallest $i$ elements in $\p(\chi')$. 
    Together, $\LexiMin((h_1,\dots,h_i)) > \LexiMin((h'_{j_1},\dots,h'_{j_i}))$, by comparing the smallest $i$ elements of  $\p(\chi)$ and  $\p(\chi)$, yielding $\LexiMin(\p(\chi))>\LexiMin(\p(\chi'))$, a contradiction.

    With $h_{i_1} \leq h'_{i_1}$ (i.e. $v_{i_1}(\chi_{i_1}) \leq v_{i_1}(\chi'_{i_1})$) in hand, since $\chi_{i_1}$ is clean, we have $v_{i_1}(\chi_{i_1}\setminus\{o_1\})  < v_{i_1}(\chi'_{i_1})$. 
    Further since $\chi_{i_1}\setminus\{o_1\}$ and $\chi'_{i_1}$ are clean (i.e. independent sets of a matroid), there exists an item $o_2 \in \chi'_{i_1} \setminus (\chi_{i_1}\setminus\{o_1\})$ such that $v_{i_1}(\chi_{i_1}\setminus\{o_1\} \cup \{o_2\}) = v_{i_1}(\chi_{i_1}) = h_{i_1}$.
    And $o_2$ is allocated to some agent $i_2 \neq i_1$ under $\chi$, as otherwise $o_2$ is not allocated to anyone and we can transfer $o_1$ from $i_1$ to $i$ and allocate $o_2$ to $i_1$ in $\chi$ violating that $\chi$ is max-USW.
    We note that $\chi'_{i_1} \setminus (\chi_{i_1}\setminus\{o_1\}) = \chi'_{i_1} \setminus \chi_{i_1}$ and $o_2 \neq o_1$ because $o_1 \notin \chi'_{i_1}$ (recall that $o_1\in \chi'_i \setminus \chi_i$) and $o_2 \in \chi'_{i_1}$.
\end{enumerate}

Repeating the same argument and letting $i_0=i$, we obtain a sequence of items and agents $(i_0,o_1,i_1,\dots,o_t,i_t)$. Let $\chi^{(k)}$ denote the allocation that transferring $o_l$ from $i_l$ to $i_{l-1}$ under $\chi$ for all $l\in[k]$. The sequence of items and agents satisfying $o_k \in 
\chi'_{i_{k-1}} \setminus \chi^{(k-1)}_{i_{k-1}}$.

See \Cref{pic:01SUB} for an illustration of the sequence. We note that the same item $o_l$ does not appear again, since for $k\geq l$, we have $o_l \in\chi^{(k)}_{i_{l-1}}$ (as a result, $o_l \notin \chi'_{i_{k-1}} \setminus \chi^{(k)}_{i_{l-1}}$) . 
Thus, the sequence must terminate when we reach the agent $i_t$ with $h_{i_t} \geq h_i+2$ or $h_{i_t} = h_i+1$ for the first time. If $h_{i_t} \geq h_i+2$, after transferring items along the sequence, we have that $\chi$ is not $\LexiMin$ optimal, a contradiction. 
If $h_{i_t} = h_i+1$, we note that agent $q$ (recall that $h_q \geq h'_q+2$) is not in the sequence since $h_q\geq h_i+2$ (inequality \eqref{ineq:q>=i+2}). 
After transferring items along the sequence, $\LexiMin(\p(\chi))$ and $h_q$ are unchanged, which means $(\chi,\chi')$ still satisfies the two conditions (i) $\chi,\chi'\in \SSA$ and (ii) $D(\p(\chi),\p(\chi'))\geq 2$,
    but the $\sum_{i \in [n]}|\chi_i \triangle \chi'_i|$ decreases, a contradiction.
    
\begin{figure}[t]
	\centering
	\includegraphics[width=.6\columnwidth]{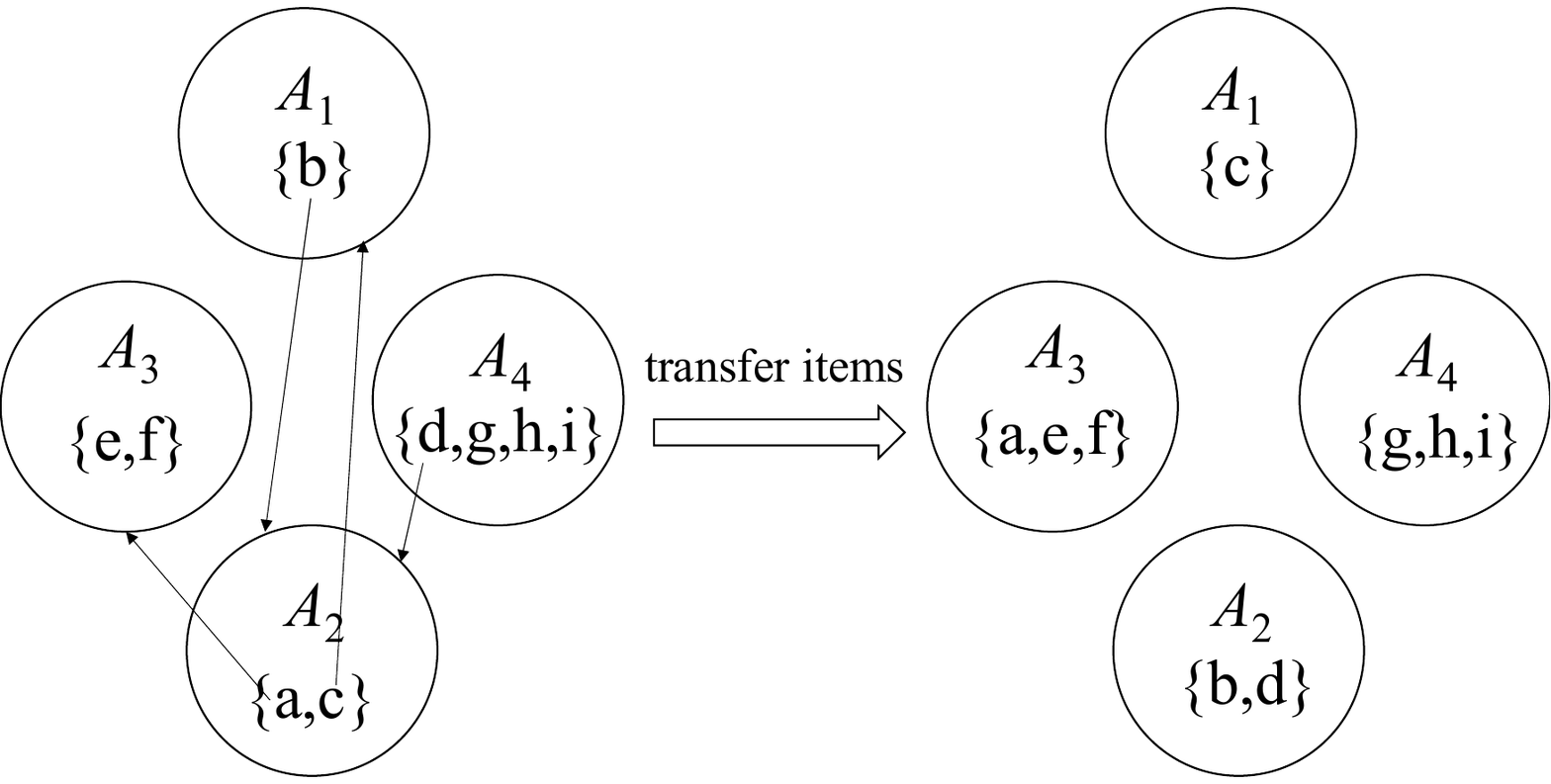}\quad
	\caption{Transferring items along the sequence $(i_0,o_1,i_1,\dots,o_t,i_t)$. Some items are denoted by $a,\dots,g$.}\label{pic:01SUB}
\end{figure}

To sum up, suppose to the opposite of \Cref{thm:01SUB-Chev} that there are some pairs of allocations $(\chi,\chi')$ that meets (i) $\chi,\chi'\in \SSA$ and (ii) $D(\p(\chi),\p(\chi'))\geq 2$, then it leads to some contradiction. 
Thus, for $\chi,\chi'\in \SSA$, it holds that $D(\p(\chi),\p(\chi'))\leq 1$.
\end{proof}

\begin{remark}
    Our proof of \Cref{thm:01SUB-Chev} is similar to that of Lemma 3.12 in \cite{10.1145/3485006}, which is the 0/1-sub valuation version of \Cref{lemma:crucial-lexmin} (implying an allocation non-optimal under $\LexiMin$ is non-optimal under any SPD criterion).
    There is a minor flaw in their proof of Lemma 3.12 in \cite{10.1145/3485006} which we have modified in our proof of \Cref{thm:01SUB-Chev}.

    In the proof of Lemma 3.12 in \cite{10.1145/3485006}, for the sequence of items and agents $(i_0,o_1,i_1,\dots,o_t,i_t)$, they argue that no same \textbf{agent} appears twice,
      which implies that if the same \textbf{agent} appears again, the allocation $\chi$ is still clean after transferring items along the cycle (which is a subsequence of the sequence). 
    However, this is not true. See the example in \Cref{pic:01SUB}, where the sequence is $(3,a,2,b,1,c,2,d,4)$.
    After transferring items along the cycle $(2,b,1,c,2)$ (which is equivalent to swapping items $b$ and $c$ between agents $2$ and $1$), the allocation may not be clean if $\{a,b\}$ is not a clean bundle of agent $2$.
    Indeed, the family of clean bundles of agent $2$ may be $\mathcal{I}_2=\{\{a,c\},\{a,d\},\{b,c\},\{b,d\},\{a\},\{b\},\{c\},\{d\}\}$ such that $\{a,b\} \notin \mathcal{I}_2$.

    We note that Lemma 3.12 in \cite{10.1145/3485006} still holds and the minor flaw in its proof in \cite{10.1145/3485006} can be corrected by arguing that no same \textbf{item} appears twice like ours proof of \Cref{thm:01SUB-Chev}.
\end{remark}
\end{sloppypar}

\section*{Summary}

For the binary valuations, it is known the set of optimums has nothing to do with the exact criterion, as long as the criterion is SPD and symmetric. A complete proof to this result was scattered in literature and is provided in this paper.
Furthermore, this paper discusses the consistency among the set of SPD optimal allocations, that is,
   their profiles are close to each other under Chebyshev distance, which holds for 0/1-sub valuations.
     The proof is nontrivial.
  
For the 0/1-add valuations, we introduce a layer partition among items and agents.
Together with the idea of dividing items into $2n^2$ identical pieces,
   we obtain the first combinatorial algorithm for computing an optimal allocation (i.e. a stable allocation) for the divisible items, which has a lower complexity than the LP approach.
  
%As for weakly Pigou-Dalton criteria (e.g. maximizing the smallest valuation), the set $\SA$ is not equal to but a subset of the set of optimums.
%However, it is still able to find an optimum efficiently via computing an element in $\SA$.

%In terms of the mixed case where there are both divisible and indivisible items, a negative result (the optimal allocations under $\LexiMax$ and $\LexiMin$ are \textbf{not} necessarily identical) is given in the supplement.

%A future work may be considering the consistency among SPD criteria and furhter onsistency among SPD optimal allocations under ternary valuations of $\{0,1,2\}$ on a single item. 

\bibliographystyle{splncs04}
\bibliography{reference}

\clearpage
\appendix

\section{All Criteria in Example~\ref{example:criteria} are SPD}\label{sect:spd-verify}

 \begin{claim}\label{claim:ES-GINI}
 Minimizing $\ES(\p)$ is equivalent to minimizing $\GINI(\p)$.
 \end{claim}

 \begin{proof}
 \begin{align*}
 	& {\sum}_{h_i<h_j}(h_j-h_i) = {\sum}_{i<j}(h^\uparrow_j-h^\uparrow_i) \\
 	=& {\sum}_{j}(h^\uparrow_j(j-1)) - {\sum}_{i}(h^\uparrow_i(n-i))\\
 	=& {\sum}_{i}(h^\uparrow_i(i-1-n+i))\\
 	=& 2{\sum}_{i}(i \cdot h^\uparrow_i) - (1+n){\sum}_{i}(h^\uparrow_i)\\
 	=& 2\GINI(\p) - (1+n)m.
 \end{align*}
 \end{proof}

 \begin{lemma}
 Consider two profiles $\p=(p_1,\ldots,p_n)$ and $\q=(q_1,\ldots,q_n)$, where
   $\q$ is more balanced that $\p$ (namely, there are $j,k\in [n]$ such that $p_j>q_j\geq q_k>p_k$ and $q_i=p_i$ for $i\in [n]\setminus \{j,k\}$).
 \begin{enumerate}
 \item $\NSW(\q) < \NSW(\p)$.
 \item $\Potential_\Phi(\q)<\Potential_\Phi(\p)$, where $\Phi(x)$ is a \textbf{strictly} convex function of $x$. 
     (Hence criteria $\sum_i \Phi(h_i)$ with strictly convex terms $\Phi(h_i)$, including $\entropy(\p)$, $\Congestion(\p)$, $\LexiMax(\p)$, and $\LexiMin(\p)$ are all SPD.)
 \item $\ES(\q) < \ES(\p)$. \\
    (Hence $\GINI(\q) < \GINI(\p)$ due to Claim~\ref{claim:ES-GINI}.)
 \end{enumerate}
 Therefore all criteria mentioned in \Cref{example:criteria} are SPD.
 \end{lemma}

 \begin{proof}
 By definition, we assume $q_j = p_j - \Delta$ and $q_k = p_k + \Delta$ where $\Delta>0$. Then we have $p_j-p_k\geq 2\Delta$.
	
 1.    \begin{align*}
 		& \NSW(\q)-\NSW(\p)  \\
 		=& -(q_jq_k-p_jp_k)\prod_{i}{p_i} \\
 		=& (p_jp_k-(p_j -\Delta)(p_k + \Delta))\prod_{i}{p_i}\\
 		=& (p_jp_k-p_jp_k - \Delta p_j + \Delta p_k + \Delta^2 )\prod_{i}{p_i}\\
 		=& (\Delta(\Delta-(p_j-p_k))\prod_{i}{p_i}\\
 		\leq& (\Delta(\Delta-2\Delta))\prod_{i}{p_i}\\
 		=&-(\Delta^2)\prod_{i}{p_i}<0.
 	\end{align*}
 	Thus we have $\NSW(\q)<\NSW(\p)$.
	
 \smallskip 2. \begin{align*}
 	& \Potential_\Phi(\q) - \Potential_\Phi(\p) \\
 	=& (\Phi(q_k) - \Phi(p_k)) + (\Phi(q_j) - \Phi(p_j))\\ 
 	=& (\Phi(p_k + \Delta) - \Phi(p_k)) - (\Phi(q_j+\Delta) - \Phi(q_j)).
 \end{align*}
 	Since $\Phi(x)$ is a strictly convex function of $x$, and  $q_j > p_k$, we have
         $$(\Phi(p_k + \Delta) - \Phi(p_k)) < (\Phi(q_j+\Delta) - \Phi(q_j)).$$
 	Thus, $$\Potential_\Phi(\q) < \Potential_\Phi(\p).$$

  %$\Potential_\Phi(\q) - \Potential_\Phi(\p) = (\Phi(p_k + \Delta) - \Phi(p_k)) + (\Phi(p_j - \Delta) - \Phi(p_j))$ is negative, which follows from $q_j < h_u - \Delta$ and the convexity of $g$. So, $P(\p)$ is Pigou-Dalton.

 \smallskip 3. Let $n_1,n_2,n_3,n_4$ and $n_5$ be the number of agents $i$ except $j,k$ with 
 	$p_i \leq p_k$, $p_k < p_i \leq q_k$,
 	$q_k < p_i < q_j$, $q_j \leq p_i < p_j$ and
 	$p_i \geq p_j$ respectively. 
	
 	Let $U(\p,i) = {\sum}_{\substack{j \\ h_i < h_j}}(h_j-h_i)$ denote the unhappiness of agent $i$. Then we have 
 	$$\ES(\p) = {\sum}_{i}U(\p,i),$$ 
 	$$\ES(\q) - \ES(\p) ={\sum}_{i}(U(\q,i)-U(\p,i)).$$
	
 	Observe the change from $\p$ to $\q$, we derive that:
 	\begin{itemize}
 		\item[(i)] For those agents $i$ with $p_i \leq p_k$ or $p_i \geq p_j$, their $U(\p,i)=U(\q,i)$s remain unchanged.
 		\item[(ii)] For a agent $i$ with $p_k < p_i \leq q_k$, its $U(\p,i)$ changes by $$U(\q,i)-U(\p,i) = (q_k-p_i + q_j-p_i) - (p_j-p_i) = q_k-p_i-\Delta \leq 0.$$
 		\item[(iii)] For those agents $i$ with $q_k < p_i < q_j$, their total sum of unhappiness is changed by $$\Delta_{3} = {\sum}_{{\substack{i \\ q_k < p_i < q_j}}}(U(\q,i)-U(\p,i))= -n_3\Delta.$$
 		\item[(iv)] For a agent $i$ with $q_j \leq p_i < p_j$, its unhappiness changes by $$U(\q,i)-U(\p,i) = -(p_j-p_i) \leq 0.$$
 		\item[(v)] For $j$, its unhappiness changes by $$\Delta_j = U(\q,j)-U(\p,j) \leq (n_4 + n_5) \Delta,$$ because of the decrease from $p_j$ to $q_j$.
 		\item[vi] For $k$, its unhappiness changes by $$\Delta_k = U(\q,k)-U(\p,k) \leq -(n_3 + n_4 + n_5) \Delta - 2 \Delta.$$ The $-(n_3 + n_4 + n_5)\Delta$ part is due to the increase from $p_k$ to $q_k$ and the remaining $-2 \Delta$ part is due to the increase from $p_k$ to $q_k$ and the decrease from $p_j$ to $q_j$.
 	\end{itemize}
 	Since case (i), (ii) and (iv) contribute non-positive terms in $\ES(\q) - \ES(\p)$, we may only consider (iii), (v) and (vi) and have:
 	\begin{align*}
 		& \ES(\q) - \ES(\p) \\
 		=&{\sum}_{i}(U(\q,i)-U(\p,i))\\
 		\leq& \Delta_{3} + \Delta_j + \Delta_k\\
 		\leq& -n_3\Delta + (n_4 + n_5) \Delta -(n_3 + n_4 + n_5) \Delta - 2 \Delta\\ 
 		=& -n_3\Delta-2\Delta < 0.
 	\end{align*}
	
 	Thus we have $\ES(\q) - \ES(\p) < 0$.
 	%So the difference between $\ES(\q)$ and $\ES(\p)$ is $\ES(\q) - \ES(\p)< 0$
 	%Further since $- n_3\Delta + (n_4 + n_5) \Delta -(n_3 + n_4 + n_5) \Delta - 2 \Delta =-(2n_3 + 2) \Delta< 0$, we see $\ES(\q) - \ES(\p)< 0$, thus $\ES$ is SPD.
 \end{proof}

\section{Inconsistency for the Mixed Case}

 When some items are divisible and others are indivisible, the optimal allocations under different SPD criteria may differ. In other words, the consistency on optimums among all the SPD criteria for divisible and indivisible case, respectively, does not generalize to the mixed case. We show an example in the following. 
 Suppose there are 4 agents and 6 items. The first 4 items are indivisible whereas items $5,6$ are divisible. The matrix $a$ below represents the preference of agents: item $i$ can be allocated to agent $j$ if and only if $a_{i,j}=1$:
 $$
 a = 
 \left(
 \begin{array} {llll}
 1 & 0 & 0 & 0 \\
 0 & 1 & 0 & 0 \\
 0 & 0 & 1 & 0 \\
 1 & 0 & 0 & 1 \\
 0 & 1 & 1 & 1 \\
 0 & 1 & 1 & 1 
 \end{array}
 \right)
 $$

 We use a matrix $b$ to represent the allocation, where $b_{i,j}$ indicates the amount of item $i$ allocated to agent $j$. Note that $b_{i,j}=0$ if $a_{i,j}=0$ and the $i$-th row of matrix $b$ contains only one 1 if item $i$ is indivisible. The sum of each row of the matrix equals 1, and the sum of each column equals the amount of items obtained by each agent.

 In this example, all allocations can be divided into two classes depending on whether item $4$ is assigned to agent $1$ or agent $4$.
 It is not hard to verify:

 The optimal allocation for minimizing $\LexiMin$ is shown in $b_{\LexiMin}$ below, admitting profile $(2,\frac{4}{3},\frac{4}{3},\frac{4}{3})$. 

 The optimal allocation for minimizing $\LexiMax$ is shown in $b_{\LexiMax}$ below, admitting profile  $(1,\frac{5}{3},\frac{5}{3},\frac{5}{3})$. 

 $$
 %\resizebox{.6\columnwidth}{!}{$
 b_{\LexiMin} = 
 \left(
 \begin{array} {llll}
 1 & 0 & 0 & 0 \\
 0 & 1 & 0 & 0 \\
 0 & 0 & 1 & 0 \\
 1 & 0 & 0 & 0 \\
 0 & \frac{1}{3} & 0 & \frac{2}{3} \\
 0 & 0 & \frac{1}{3} & \frac{2}{3} 
 \end{array} \right),  b_{\LexiMax} = 
 \left( 
 \begin{array} {llll}
 1 & 0 & 0 & 0 \\
 0 & 1 & 0 & 0 \\
 0 & 0 & 1 & 0 \\
 0 & 0 & 0 & 1 \\
 0 & \frac{1}{3} & \frac{1}{3} & \frac{1}{3} \\
 0 & \frac{1}{3} & \frac{1}{3} & \frac{1}{3} 
 \end{array}
 \right).
 %$}
 $$

\end{document}